\documentclass[journal,10pt]{IEEEtran}
\usepackage{cite}

\ifCLASSINFOpdf
  \usepackage[pdftex]{graphicx}
\else
   \usepackage[dvips]{graphicx}
  \graphicspath{{./fig/}}
  \usepackage[usenames,dvipsnames]{pstricks}
\usepackage{pst-grad} 
\usepackage{pst-plot} 
\fi
\usepackage{amssymb}
\usepackage{latexsym}

\usepackage[cmex10]{amsmath}
\usepackage{amsfonts}

\usepackage{algorithmic}

\usepackage{rotating}

\usepackage{array}

\usepackage{mdwmath}
\usepackage{mdwtab}
\usepackage[caption=false,font=footnotesize]{subfig}
\usepackage{fixltx2e}

\usepackage{stfloats}

\usepackage[on]{auto-pst-pdf}

\usepackage{url}

\usepackage{array,setspace}

\usepackage{amsthm,multirow,footnote,hyperref}
\usepackage[on]{auto-pst-pdf}
\graphicspath{{./fig/}}
\usepackage{url}
\usepackage{booktabs}
\renewcommand{\arraystretch}{1.2}
\newcommand{\ra}[1]{\renewcommand{\arraystretch}{#1}}

\usepackage{algorithmic}
\usepackage{algorithm}
\makesavenoteenv{tabular}
\providecommand{\eqdef}{\overset{def}{=}}

\providecommand{\Cspace}{\mathbb{C}}
\providecommand{\Rspace}{\mathbb{R}}

\providecommand{\Zspace}{\mathbb{Z}}
\providecommand{\Nspace}{\mathbb{N}}

\providecommand{\Esp}[1]{\mathbb{E}\left[#1\right]}

\providecommand{\mat}[1]{\boldsymbol{#1}}
\providecommand{\vct}[1]{\boldsymbol{#1}}

\providecommand{\pinv}{\dag}
\providecommand{\krop}{\otimes}

\providecommand{\xp}[1]{\textsf{Exp.~{#1}}\ }

\providecommand{\pinv}{\text{\ding{61}}}

\newtheorem{proposition}{Proposition}

\newtheorem{corollary}{Corollary}
\newtheorem{lemma}{Lemma}

\begin{document}
%
\title{Estimation of Sparse MIMO Channels with Common Support.}
%
%
%

\author{Yann~Barbotin,~\IEEEmembership{Student~Member,~IEEE,}
        Ali~Hormati,~\IEEEmembership{Member,~IEEE,}
        Sundeep~Rangan,~\IEEEmembership{Member,~IEEE,}
        and~Martin~Vetterli,~\IEEEmembership{Fellow,~IEEE}
\thanks{Y. Barbotin, A. Hormati and M. Vetterli are with Ecole Polytechnique F\'{e}d\'{e}rale de Lausanne, 1015 Lausanne, Switzerland.}
\thanks{S. Rangan is with the Department of Electrical and Computer Engineering, Polytechnic Institute of New York University, Brooklyn, NY.}%
\thanks{M. Vetterli is also with the Department of
Electrical Engineering and Computer Sciences, University of
California, Berkeley, CA.}%
\thanks{This work has been submitted to the IEEE for possible publication. Copyright may be transferred without notice, after which this version may no longer be accessible.}%
\thanks{This research is supported by \emph{Qualcomm Inc.}, \emph{ERC Advanced Grant – Support for Frontier Research - SPARSAM Nr : 247006} and \emph{SNF Grant - New Sampling Methods for Processing and Communication Nr : 200021-121935/1}.}}%


%
%

\markboth{Barbotin \MakeLowercase{\textit{et al.}}: Estimation of Sparse MIMO Channels with Common Support.}{}%
%



\maketitle

\begin{abstract}
We consider the problem of estimating sparse communication channels in the MIMO context.  In small to medium bandwidth communications,  as in the current standards for OFDM and CDMA communication systems (with bandwidth up to 20 MHz), such channels are individually sparse and at the same time share a common support set. Since the underlying physical channels are inherently continuous-time, we propose a parametric sparse estimation technique based on finite rate of innovation (FRI) principles. Parametric estimation is especially relevant to MIMO communications as it allows for a robust estimation and concise description of the channels.

The core of the algorithm is a generalization of conventional spectral estimation methods to multiple input signals with common support. We show the application of our technique for channel estimation in OFDM (uniformly/contiguous DFT pilots) and CDMA downlink (Walsh-Hadamard coded schemes). In the presence of additive white Gaussian noise, theoretical lower bounds on the estimation of SCS channel parameters in Rayleigh fading conditions are derived. Finally, an analytical spatial channel model is derived, and simulations on this model in the OFDM setting show the symbol error rate (SER) is reduced by a factor 2 (0 dB of SNR) to 5 (high SNR) compared to standard non-parametric methods --- e.g. lowpass interpolation.

\end{abstract}
%
\section{Introduction}

\begin{figure*}[h]
\centering
\includegraphics[width=\textwidth]{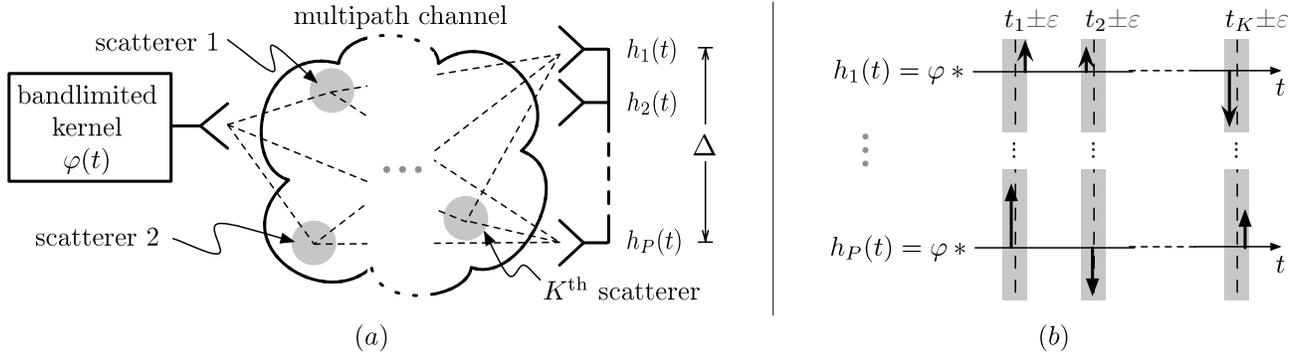}
\caption{(a) Transmission over a bandlimited medium with two scatterers and $P$ receiving antennas. (b) The $P$ channels contain two paths arriving at the same time up to $\pm\varepsilon$, and are thus no exact SCS channels for $\varepsilon>0$.}
\label{fig:SCSMIMO}
\end{figure*}

Multiple input multiple output (MIMO) antenna wireless systems
enable significant gains in both throughput and reliability
\cite{Telatar1999,Foschini1998,Alamouti1998,Zheng2003} and are now incorporated
in several commercial wireless standards \cite{IEEE802.11n,Dahlman2008}.
However, critical to realizing the full potential of
MIMO systems is the need for accurate channel estimates at the receiver,
and, for certain schemes, at the transmitter as well.
As the number of transmit antennas is increased,
the receiver must estimate proportionally  more channels,  which in turn increases the
pilot overhead and tends to reduce
the overall MIMO throughput gains \cite{HochwaldHass2003}.

To reduce this channel estimation overhead, the key insight of this paper is
that most MIMO channels have an approximately \emph{sparse common support} (SCS).
That is, the channel in each transmit-receive (TX-RX) antenna pair can be modeled as a discrete
multipath channel, with the relative time delays being common across different TX-RX pairs.
The commonality across the different antenna pairs reduces the overall number of degree
of freedom to estimate, which can in turn be used to reduce the pilot overhead or
improve the channel estimate.  Also, in communication systems that depend on channel state
feedback from the RX to the TX, the SCS model may enable a more compressed representation.

To exploit the SCS property of MIMO channels, we propose a variant along \cite{barbotin2011a,theLetterPaper} of the
finite rate of innovation (FRI) framework, originally developed in \cite{Blu2008}.
The method, which we call SCS-FRI, uses
classical spectral estimation techniques such as Prony's method, ESPRIT and Cadzow denoising
to recover the delay positions in frequency domain.  The method is computationally simple,
and our simulations demonstrate excellent performance in practical scenarios.
The prposed SCS-FRI algorithm applies immediately to channel estimation in multi-output OFDM communication with contiguous or uniformly scattered DFT pilots. Interestingly enough it can be used on other modulation schemes provided a suitable pilot layout. The Walsh-Hadamard transform (WHT), used in CDMA downlink channel among others, qualifies if one controls the pilots layout in the WHT domain.

We also derive a simple scalar formula for the Cram\'er-Rao bound on the estimation of separable ToAs,
and also point to a more general result by Yau and Bresler \cite{Yau1992}. Both bounds are extended to Rayleigh fading SCS channels to lower bound the expected estimation error in fading conditions.
Our simulations indicate the proposed SCS-FRI method is close to this bound at high SNRs.


\subsection{SCS MIMO models}

Due to the physical properties of outdoor electromagnetic propagation, wireless channels are often
modeled as having a channel impulse response (CIR) that is
sparse in the sense that they contain few significant paths \cite{Rappaport2002}.
With multiple antennas, the CIR measured at different antennas share a common support, i.e. the times of arrival (ToA) at different antennas are similar while the paths amplitudes and phases
are distinct. This sparse common support (SCS) channels is illustrate it in Figure~\ref{fig:SCSMIMO}. The SCS channel model is usually assumed in the literature, though its physical motivations are not always put forth.

It is important to note that the sparsity and common support assumptions only hold with respect to the channel bandwidth $B$ and the SNR of the channel. Indeed, in the presence of noise, resolution is limited by the inverse bandwidth $1/B$, even if one knows exactly which parametric model the signal obeys. In practice, $1/10^{\text{th}}$ of the inverse bandwidth is a reasonable resolution to shoot for. The limited resolution has the effect of clustering paths from a single scatterer into a single path (promoting sparsity), and the small shift in the ToA due to the distance between antennas becomes negligible (promoting common support). Table~\ref{tabl:invB} gives the channel bandwidth of several modern standards and $c/B$ which is the distance travelled by an electromagnetic wave in a time lapse equal to the inverse bandwidth.

\begin{table}[t]\centering  \ra{1.3}\caption{Channel bandwidth in popular wireless systems}
\begin{tabular}{@{}llrr@{}}\toprule
\textbf{System} & \textbf{Code}& \textbf{Bandwidth} $B$ & \begin{array}{r}\textbf{Resolvable}\vspace{-8pt}\\ \textbf{distance } $c/B$\end{array} \\\midrule
DVB-T \cite{ETSI2001}& DFT & $5$--$8$ MHz & $38$--$60$ m\\
IS-95  \cite{Garg2000}& WHT& $1.25$ MHz & $240$ m\\%
3GPP LTE \cite{Dahlman2008}& DFT & $1.4$--$20$ MHz & $15$--$215$ m\\   %
UWB & --- &  $>500$ MHz & $<60$ cm\\
\bottomrule
\end{tabular} \label{tabl:invB}\vspace{-4mm}
\end{table}

\subsection{Related work}
In OFDM systems, the majority of commercial channel estimators often simply perform 
some form of linear filtering or interpolation of the pilot symbols \cite{RanaKC:10,KarakayaAC:08}.
Such non-parametric techniques are computationally very simple, but fundamentally 
cannot exploit the common sparsity in MIMO channel models. Since the phases and magnitudes are generally independent on the paths on different antenna pairs, the frequency response of sparse common support (SCS) channels are not correlated in any simple manner that can be exploited by basic linear interpolation of pilots.

A different line of work has proposed compressed sensing based 
methods for sparse channel estimation \cite{ParedesAW:07,TaubockH:08,HauptBRN:10,Schniter:10}.
In the compressed sensing context, the SCS property is equivalent to \emph{joint} or
\emph{group sparsity} for which there are several methods including group LASSO 
\cite{YuanLin:06,ZhaoRY:08}, group OMP \cite{LozanoSA:08} and belief propagation
\cite{KimCJBY:11}.  Techniques for mixes of joint and individual sparsity are considered in
\cite{Wakin2005,Duarte2005}. All of these compressed sensing methods, however, require
that the delay locations are discretized and exact sparsity is achieved only when the true path
locations fall exactly on one of the discrete points.  With continuous value
path locations, each path components will require a number of terms to approximate well,
or demand a larger number of dictionary elements to offer a finer discretization.

Another joint estimation problem with FRI signals is studied in \cite{Gedalyahu2010}.
\begin{table}[t!]\centering  \ra{1.3}\caption{Channel estimation methods are naturally classified in terms of the channel properties they exploit. }
\begin{tabular}{@{}c@{ }cr@{}cccc@{}}\toprule
& & \textbf{Algorithm} & \phantom{a} &  \multicolumn{3}{c}{\textbf{Exploited channel properties}} \\\cmidrule{5-7}
& & & & Short & & Common \\
& & & & delay-spread & Sparsity & support\\\midrule
\multirow{7}{*}{\centering\begin{sideways}\textbf{DFT pilots layout}\end{sideways}}&\multirow{3}{*}{\centering\begin{sideways}Contiguous\ \end{sideways}} & Lowpass & & \checkmark & &  \\
& &FRI & & & \checkmark &  \\
& &SCS-FRI & & & \checkmark & \checkmark \\ \\
& \multirow{3}{*}{\centering\begin{sideways}\centering Scattered\ \ \end{sideways}} &Lowpass & & \checkmark & &  \\
& &FRI & &\checkmark & \checkmark &  \\
& &SCS-FRI & &\checkmark & \checkmark & \checkmark \\
\bottomrule
\end{tabular}\vspace{-4mm} \label{tabl:prop}
\end{table}
\subsection{Contributions}
The contributions of this work are four-folds:
\begin{itemize}
\item Extension of classical FRI sampling and estimation to multiple SCS channels (Section~\ref{sec:theory})
\item Derivation of simple scalar formulas for the CRB of SCS channels (Section~\ref{sec:CRB})
\item Application to OFDM and Walsh-Hadamard coded (e.g. CDMA downlink) communications with contiguous or uniformely scattered DFT pilots (Section~\ref{sec:ofdm})
\item Characterization of a precise spatial analytical model for SCS channels (Section~\ref{sec:application})
\end{itemize}

The proposed SCS-FRI algorithm stands out compared to FRI or lowpass interpolation as it exploits more channel properties, as indicated in Table~\ref{tabl:prop}. Lowpass based techniques are a sensible non-parametric approach as they exploit the short delay-spread property. In general, any estimation technique based on uniformly scattered DFT pilots uses this property, as it is a necessary condition to the unicity of the solution.

We conclude our study with numerical simulations showing the efficiency of the SCS-FRI algorithm in a Rayleigh fading scenario, and compare its equalization gain to a standard non-parametric approach, i.e. lowpass interpolation in the DFT domain.

%
%
%
%

\section{Sparse Common Support FRI: Theory and Algorithms}\label{sec:theory}

\subsection{Problem formulation}
We consider the physical setup described in Figure~\ref{fig:SCSMIMO}.(a). A periodic signal of limited bandwidth is transmitted over a multipath channel and uniformly sampled by a receiver with $P$ antennas. This leads to $P$ parallel multipath channels as shown in Figure~\ref{fig:SCSMIMO}.(b). The channels either share a common support exactly, in which case they are called exact SCS channels, or approximately, in which case they are called SCS channels (e.g. Figure~\ref{fig:SCSMIMO}.(b)).

Consider $P$ exact SCS channels shaped by a kernel $\varphi$, with complex baseband equivalent model:
\begin{equation}\label{eq:mpchannelth}
 h_p(t)=\sum_{k=1}^K c_{k,p}\varphi(t-t_k)\ ,\quad c_{k,p}\in\Cspace,\ t_k\in [0\ \tau[\ ,
\end{equation}
where $\varphi(t)$ is the $\tau$-periodic sinc function or Dirichlet kernel:
\begin{equation}
\label{chap2equ:DirichletKernel}
  \varphi(t) = \sum_{k \in \mathbb{Z}} \text{sinc}(B(t-k\tau)) = \frac{\sin(\pi B t)}{B\tau \sin(\frac{\pi t}{\tau})}\,.
\end{equation}

The kernel $\varphi$ is considered periodic as the filtering of a periodically padded signal by a linear shift invariant filter. Therefore, linear convolution of the CIR with the shaping kernel becomes circular.  

We assume that the bandwidth parameter $B$ satisfies $B = (2M+1)/\tau$ for $M \geq K$. The paths coefficients $c_{k,p}$ are treated as complex random variables. $N$ measurements $y_p[n]$ are acquired at a rate $1/T=N/\tau$ (with $\tau$ the signal period and $N\geq B\tau=2M+1$) and corrupted by AWGN
\begin{equation}y_p[n]=h_p[n]+q_p[n]\ \quad n\in\lbrace 0,\dots,N-1\rbrace , \label{eq:spl}\end{equation}
 where $\vct q_p \sim\mathcal N_\Cspace(\vct 0,\sigma^2\mathbb I)$ if the channel is complex-valued or $\vct q_p \sim\mathcal N(\vct 0,\sigma^2\mathbb I)$ if real-valued.
In the DFT domain, the received signal is:
\begin{equation}\widehat{y}_p[m]=\widehat{\varphi}[m]\cdot\sum_{k=1}^K c_{k,p}W^{mt_k}+\widehat{q}_p[m].\label{eq:spldft}\end{equation}
where $W=e^{-2\pi j /\tau}$ and $\widehat{\varphi}[m] = 1/(2M+1) $ for $|m| \leq M$ and is zero otherwise. The goal is to estimate the support $\lbrace t_k\rbrace_{k=1\dots K}$ and the paths amplitudes $\lbrace c_{k,p}\rbrace_{k=1\dots K,p=1\dots P}$ from the $NP$ samples collected in (\ref{eq:spl}). Once the support is known, estimation of the path amplitudes is simple linear algebra as seen in (\ref{eq:spldft}).

\subsection{Support recovery from baseband DFT coefficients}
We start from (\ref{eq:spldft}). The DFT samples $\widehat{y}_p[m]$ in the baseband ($|m|\leq M$) are the DFT coefficients of the channel corrupted by some Gaussian noise.

The noiseless DFT coefficients of a $K$-multipath channel have the well-known and interesting property to form a linear recurrent sequence of order $K+1$, i.e., any coefficient $\widehat h_p[m]$ ($m\geq -M+K$) can be expressed as a unique linear combination of the $K$ previous DFT coefficients common to all indices $m$:

\begin{lemma}\label{lemma:recseq}
Given $\widehat{h}_p[m]=\sum_{k=1}^K c_{k,p}W^{mt_k}$ for $m=-M+K,\dots, M$ and $t_i\neq t_j,\ \forall i\neq j$,
there exists  a unique set of coefficients $\lbrace f_k\rbrace_{k=1,\dots,K}$ such that:
$$\widehat{h}_p[m]=f_1\widehat{h}_p[m-1]+f_2\widehat{h}_p[m-2]+\dots+f_K\widehat{h}_p[m-K]$$
where $x^K-f_1x^{K-1}-\dots -f_{K-1}x-f_K$ is the polynomial with roots $\lbrace W^{t_k}\rbrace_{k=1,\dots,K}$.
\end{lemma}
\begin{proof}
A linear recursion of degree $K$ can be written as:
\begin{equation}x_n=f_1x_{n-1}+\dots +f_Kx_{n-K},\ f_K\neq 0.\label{eq:linrec}\end{equation}
Its characteristic equation is:
\begin{equation}x^K-f_1x^{K-1}-\dots -f_{K-1}x-f_K=0.\label{eq:chareq}\end{equation}
If $\lambda_x$ is a solution of (\ref{eq:chareq}) then multiplying both sides of the equation by $\lambda_x^{n-K}$ ($\neq 0$ since $f_K\neq 0$) shows that $\lambda_x^{n}$ is a solution of (\ref{eq:linrec}).
Moreover by linearity, any linear combination of solutions of (\ref{eq:linrec}) is still a solution, and if (\ref{eq:chareq}) has $K$ distinct solutions, $\lbrace f_k\rbrace_{k=1,\dots,K}$ is uniquely defined by a set of $K$ independent linear equations.
Hence, for $\sum_{k=1}^K c_{k,p}W^{mt_k}$ ``solution'' of (\ref{eq:linrec}), $t_k\not\equiv t_l \text{ mod }\tau$ for all $\ k\neq l$, there exists a unique set $\lbrace f_k\rbrace_{k=1,\dots,K}$ such that $\lbrace W^{t_k}\rbrace_{k=1,\dots,K}$ are the $K$ distinct roots of
$$x^K-f_1x^{K-1}-\dots -f_{K-1}x-f_K.$$
\end{proof}

The coefficients $\widehat{y}_p[m]$ maybe arranged in a tall block-Toeplitz matrix 
$$\mat H^{(L)}=\begin{bmatrix}\mat H_1^{(L)}\\ \mat H_2^{(L)}\\ \vdots \\ \mat H_P^{(L)}\end{bmatrix}$$
 such that
\begin{equation}\mat H_p^{(L)}=\begin{bmatrix}
\widehat y_{p,L-M-1} &\widehat y_{p,L-M-2} & \cdots &\widehat y_{p,-M}\\
\widehat y_{p,L-M} &\widehat y_{p,L-M-1} & \cdots &\widehat y_{p,1-M}\\
\vdots &\vdots & \ddots & \vdots \\
\widehat y_{p,M} &\widehat y_{p,M-1} & \cdots &\widehat y_{p,M-L+1}
\end{bmatrix},\label{eq:H}
\end{equation}
where $\widehat y_{i,j}=\widehat y_i[j]$. The \emph{data matrix} $\mat H^{(L)}$ is made of $P$ Toeplitz blocks of size $(2M+2-L)\times L$, and we assume that both block dimensions are larger or equal to $K$. It possess interesting algebraic properties which form the core of line spectra estimation techniques. We will use Lemma~\ref{lemma:recseq} to show three well-known spectral estimation tools which extend straightforwardly from Topelitz data matrices to block-Toeplitz ones, i.e. extend from single output to multiple outputs with SCS. We do so, with two propositions:
\begin{proposition}\label{prop:af}\textbf{[Annihilating filter property]}\\
In the absence of noise ($\widehat y_p[m]=\widehat h_p[m]$), a set of exact SCS channels with $K$ distinct paths verifies
\begin{equation}\mat H^{(K+1)} \vct f =\vct 0,\label{eq:af}\end{equation}
where $\vct f=[1\ -f_1\ \cdots\ -f_K]^T$ are the annihilating filter coefficients such that the polynomial
$p_{\vct f}(z)=1-\sum_{k=1}^K f_k z^{k}$
has $K$ roots $\lbrace e^{-2\pi j t_k/\tau}\rbrace_{k=1\dots K}$. The matrix $\mat H^{(K+1)}$ is built with blocks as in (\ref{eq:H}) (with $L=K+1$).
\end{proposition}
\begin{proof}
This is a direct consequence of Lemma \ref{lemma:recseq}.
\end{proof}
\begin{proposition}\label{prop:lr}
For a set of exact SCS channels with $K$ distinct paths and in the absence of noise, $\mat H^{(L)}$ satisfies
$$\text{rank}\ \mat H^{(L)}= K,$$
for $K\leq L\leq 2M+2-K$.
\end{proposition}
\begin{proof}
Let $\widetilde{\mat H}$ be the top-left $K\times K$ minor of $\mat H^{(L)}$. It can be written as the sum of $K$ rank-1 matrices:
$$\widetilde{\mat H}=\sum_{k=1}^K c_{k,\cdot} W^{(M+1-L)t_k}\vct \xi_k^\ast\vct\xi_k,$$
such that $\vct \xi_k=[1\ W^{t_k}\ W^{2t_k}\ \cdots \ W^{(K-1)t_k}]$. If $t_k\not\equiv t_l \text{ mod }\tau$ for all $\ k\neq l$, then $\lbrace\vct\xi_1,\ \dots,\ \vct\xi_K\rbrace$ form a set of non colinear vectors. Therefore  $\text{rank}\ \mat H^{(L)} \geq K$.

Choose $K$ mutually independent rows of $\mat H^{(L)}$. From  Lemma \ref{lemma:recseq}, truncating these row vectors to length $K$ preserves the linear independence. Therefore, given a row $\vct h =[h[0],\dots,h[L-1]]$ of $\mat H^{(L)}$, there exists a linear combination of these $K$ rows $\vct h^{\prime}$ such that
$$h[k]=h^\prime[k]\ , \quad k=0,\dots,K-1.$$
By construction, since all rows of $\mat H^{(L)}$ verify the the same linear recursion of degree $K$, $\vct h$ and $\vct h^\prime$ verifies this recursion. Hence $\vct h=\vct h^\prime$, implying $\text{rank}\ \mat H^{(L)}$ is at most $K$.
\end{proof}
\paragraph{Block-Prony algorithm} 
Proposition~\ref{prop:af} is Prony's method \cite{Prony1795,Tufts1982} for block-Toeplitz matrices.
We call the corresponding algorithm ``Block-Prony TLS'', listed under Algorithm~\ref{algo:prony}. It solves the annihilating filter equation (\ref{eq:af}) in the total least-square (TLS) sense. The crucial step is the identification of what shall be the unidimensional null space of $\mat H^{(K+1)}$ in a noiseless case. Solving this problem in the TLS sense yields the least right singular vector of $\mat H^{(K+1)}$. Prony's method is notoriously sensitive to noise, which is to be expected as the result relies on identification of the unidimensional complement of the $K$-dimensional signal space. This sensitivity can be mitigated with prior denoising of the measurements.

\begin{algorithm}[t]
\caption{Block-Prony\_TLS}
\begin{algorithmic}[1]
\REQUIRE An estimate on the number of effective paths $K^{\text{est}}$, $2M+1\ (M\geq K)$ channel DFT coefficients $\widehat y_p[m]=\sum_{k=1}^K c_{k,p}W_N^{mt_k}+ \widehat q_p[m]$ for $|m|\leq M$, $p=1\dots P$.
\STATE Build $\mat H^{(K^{\text{est}}+1)}$ according to (\ref{eq:H}).
\STATE Compute the SVD decomposition of the data matrix: $\mat H^{(K^{\text{est}}+1)}=\mat U\mat S\mat V^\ast$.
\STATE $\vct \phi \leftarrow\text{roots}(\vct v)$, such that $\vct v$ is the right singular vector associated to the least singular value.
\RETURN $\lbrace t_k^{\text{est}}\rbrace_{k=1\dots K^{\text{est}}} \leftarrow -\frac{\tau}{2\pi}\text{arg}\:\vct\phi$.
\end{algorithmic}\label{algo:prony}
\end{algorithm}
\begin{algorithm}[t]
\caption{Block-ESPRIT\_TLS}
\begin{algorithmic}[1]
\REQUIRE An estimate on the number of effective paths $K^{\text{est}}$, $2M+1\ (M\geq K)$ channel DFT coefficients $\widehat y_p[m]=\sum_{k=1}^K c_{k,p}W_N^{mt_k}+ \widehat q_p[m]$ for $|m|\leq M$, $p=1\dots P$.
\STATE Build $\mat H^{(M)}$ according to (\ref{eq:H}).
\STATE Compute the SVD decomposition of the data matrix: $\mat H^{(M)}=\mat U\mat S\mat V$.
\STATE Extract the signal subspace basis $\mat\Xi_0=\mat V_{1:(M-1),1:K^{\text{est}}}$.
\STATE Extract the rotated signal space basis $\mat\Xi_1=\mat V_{2:M,1:K^{\text{est}}}$.
\STATE Solve $\mat \Xi_1=\mat\Xi_0\mat\Psi$ in the TLS sense.
\RETURN $\lbrace t_k^{\text{est}}\rbrace_{k=1\dots K^{\text{est}}} \leftarrow \lbrace -\frac{\tau}{2\pi}\text{arg}\:\lambda_k\left(\mat\Psi\right)\rbrace_{k=1\dots K^{\text{est}}}$.
\end{algorithmic}\label{algo:esprit}
\end{algorithm}
\begin{algorithm}[t!]
\caption{Block-Cadzow denoising}
\begin{algorithmic}[1]
\REQUIRE A block-Toeplitz matrix $\mat H^{(L)}$ and a target rank $K$.
\ENSURE  A block-Toeplitz matrix $\mat H^{(L)}$ with rank $\leq K$.
\REPEAT
\STATE Reduce $\mat H^{(L)}$ to rank $K$ by a truncated SVD.
\STATE Make $\mat H_p^{(L)}$ $p=1\dots P$, Toeplitz by averaging diagonals.
\UNTIL{convergence}
\end{algorithmic}\label{algo:blockcadzow}
\end{algorithm}
\begin{algorithm}[h!]
\caption{SCS-FRI channel estimation}
\begin{algorithmic}[1]
\REQUIRE An estimate on the number of effective paths $K^{\text{est}}$, $2M+1\ (M\geq K)$ noisy channel DFT coefficients $\widehat y_p[m]=\sum_{k=1}^K c_{k,p}W_N^{mt_k}+ \widehat q_p[m]$ for $|m|\leq M$, $p=1\dots P$.
\ENSURE  Support estimate $\lbrace t_k^{\text{est}}\rbrace_{k=1\dots K^{\text{est}}}$
\STATE Build $\mat H^{(M)}$ according to (\ref{eq:H}).
\STATE $\mat H^{(M)}\leftarrow \text{Block-Cadzow}(\mat H^{(M)},K^{\text{est}})$ \emph{[optional]}.
\STATE Update $\widehat y_p[m]$ with the first row and column of the denoised block $\mat H_p^{(M)}$.
\STATE Estimate the common support with Block-Prony\_TLS or Block-ESPRIT\_TLS.
\STATE Estimate $\lbrace c_{k,p}\rbrace$  solving $P$ linear Vandermonde systems (\ref{eq:spldft}).
\end{algorithmic}\label{algo:csfri}
\end{algorithm}
\paragraph{Block-ESPRIT algorithm} 
Proposition \ref{prop:lr} implies each block in the data matrix shares the same signal subspace. Hence the ESPRIT TLS algorithm outlined in \cite{Roy1989} applies as-is to the block-Toeplitz data matrix $\mat H^{(L)}$. The Block-ESPRIT TLS algorithm is outlined in Algorithm \ref{algo:esprit}.
 The Block-ESPRIT algorithm fulfils the same goal as the Block-Prony algorithm, but its essence is entirely different. Where Prony's method identifies a line with least energy in a $K$ dimensional space, ESPRIT finds the rotation between two $K$-dimensional subspaces in an $M$-dimensional space. What makes ESPRIT much more resilient to noise is that the two subspaces are computed from the most energetic part of the signal.

\paragraph{Block-Cadzow denoising} 
Proposition \ref{prop:lr} used together with the block-Toeplitz structure property yields the ``\emph{lift-and-project}'' denoising Algorithm \ref{algo:blockcadzow}, which we call  \emph{Block-Cadzow} denoising \cite{Hormati2010}. Using the same argument as in \cite{Cadzow1988}, the block-Cadzow algorithm provably converges. 

\paragraph{SCS-FRI}
We have all the elements to describe the SCS-FRI algorithm. 
The Block-Cadzow algorithm may be used to denoise the measurements and is followed by either Block-Prony or Block-ESPRIT estimation of the common ToAs (solved in the TLS sense).

\begin{figure*}[t]
\centering
\input{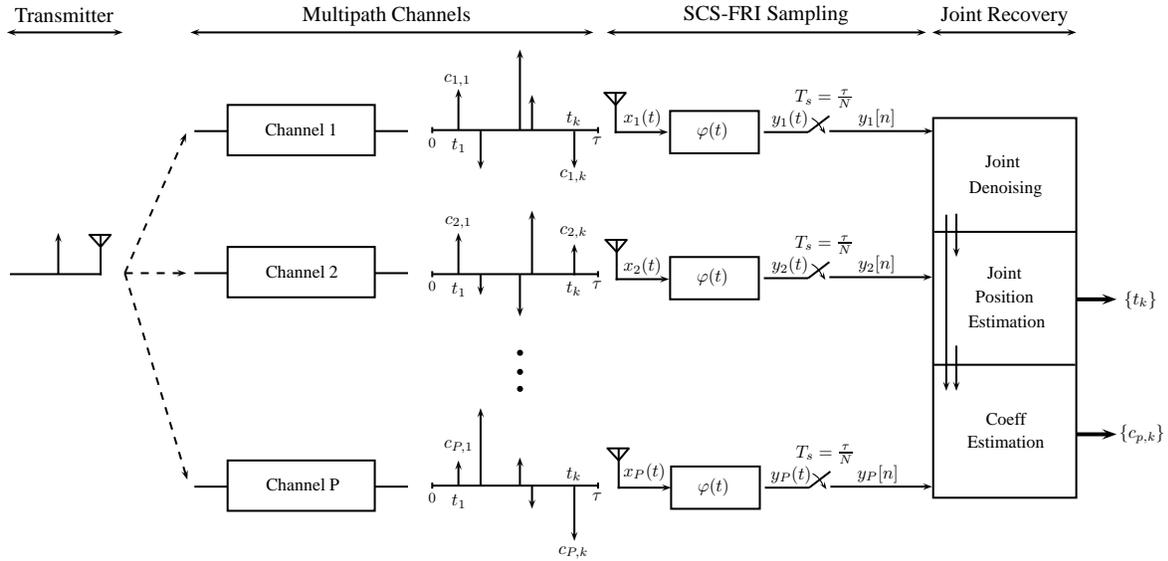}
\caption[]{The SCS-FRI sampling and reconstruction scheme in a multi-antenna channel estimation setting with $P$ receiving antennas.}\label{fig:sysbd}
\end{figure*}

For Cadzow denoising and ESPRIT, it is empirically found that a data matrix with square blocks works well.
The last step is to estimate the path amplitudes independently for each channel. This is done by solving a linear Vandermonde system (\ref{eq:spldft}) \cite{Vetterli2002}.
The processing chain at the receiver is listed in Algorithm \ref{algo:csfri} and shown in Figure~\ref{fig:sysbd}. Combination of Cadzow and ESPRIT for estimation of a single OFDM channel is considered in \cite{Prabhu2009}.
We assumed the number of paths to be known. For estimation techniques of the number of paths, we refer to \cite{Roy1989}.

\section{Estimation theoretic bounds on SCS-FRI recovery}\label{sec:CRB}
\subsection{Deterministic multipath channels}

In~\cite{Blu2008}, the authors derive the Cram\'{e}r-Rao lower bound~\cite{Cramer1946, Rao1945} for estimating the positions and weights of the Diracs in FRI signals. Considering a single Dirac with deterministic amplitude in a single-channel real-valued scenario, the minimal relative uncertainties on the location of the Dirac, $t_1$, and on its amplitude, $c_1$, are given by
\begin{align*}
\label{chap2equ:singleframeCRB}
  \Esp{\left(\frac{\vartriangle\!\! t_1}{\tau}\right)^2}&\geq \frac{3(2M+1)}{4 \pi^2 N M(M+1)}\:\mathsf{PSNR}^{-1} \\
  \Esp{\left(\frac{\vartriangle\!\! c_1}{c_1}\right)^2}&\geq \frac{2M+1}{N}\:\mathsf{PSNR}^{-1} \nonumber
\end{align*}
where $\mathsf{PSNR} = c_1^2/\sigma^2$ is the input peak signal to noise ratio. When there are more than two Diracs, the Cram\'{e}r-Rao formula for one Dirac still holds approximately when the locations are sufficiently far apart\footnote{Empirically, the distance should be larger than $2/B$.}.

\subsection{Jointly Gaussian multipath channels}
We derive bounds on the support estimation accuracy with measurements taken according to (\ref{eq:spl}). 
The paths coefficients $c_{k,p}$ are assumed to be jointly Gaussian, and modeled as the product of $a_{k,p}=\Esp{\left|c_{k,p}\right|}$ by a standard normal random variable $Z_{k,p}$  having the following properties, consistent with the well-known Rayleigh-fading model:
\begin{itemize}
\item $Z_{k,p}\sim\mathcal N_\Cspace(0,\sqrt{1/2}\mathbb I)$.
\item Similar expected path amplitude between antennas: $a_k \eqdef a_{k,1}=a_{k,2}=\dots=a_{k,P}.$
\item Independence between paths: $\Esp{Z_{k,p}Z_{k^\prime,p^\prime}^\ast}=0,\ k\neq k^\prime.$
\item The random vector $\vct Z_k=[Z_{k,1} \cdots Z_{k,P}]^T$ is defined as
$\vct Z_k=\mat L_k \vct r$,
where $\mat L_k$ is the Cholesky factor of the covariance matrix $\mat R_k=\Esp{\vct Z_k\vct Z_k^\ast}$ and $\vct r$ is a vector of iid standard complex Gaussian random variables.
\end{itemize}
The Rayleigh-fading case can be seen as deterministic  if conditioned  on the path amplitudes. Thus, the Cram\'er-Rao bounds for random paths coefficients are random variables for which we can compute statistics. Expectation and standard deviation will respectively give the expected accuracy of the estimator and its volatility.
For a single path, and a symmetric or antisymmetric $\varphi$ (not necessarily a sinc kernel), the Cr\'amer-Rao bound has a concise closed form formula:
\begin{proposition}\label{prop:crbp}
With complex-valued measurements according to (\ref{eq:spl}), $K=1$, and $\vct Z_1$ be a random Gaussian vector , then
\begin{equation}\Esp{(\frac{\vartriangle\!\! t_1}{\tau})^2}\geq\frac{\Esp{\left(\vct Z_1^\ast \vct Z_1\right)^{-1}}}{2N\cdot \textsf{dSNR}},\end{equation}
where $\textsf{dSNR}=|a_1|^2\|\varphi^\prime(\vct nT-t_1)\|^2/(N\sigma^2)$ is the differential SNR and $\lambda_1,\cdots,\lambda_P$ are the eigenvalues of the covariance matrix $\mat R_1$ and $P>1$:
\begin{itemize}
\item Uncorrelated paths coefficients, $\lambda_1=\cdots=\lambda_P=1$:
\begin{equation*}
\Esp{\left(\vct Z_1^\ast \vct Z_1\right)^{-1}}=(P-1)^{-1}\ ,\quad P>1.
\end{equation*}
\item Correlated path coefficients, such that $\lambda_1\neq\cdots\neq\lambda_P$:
\begin{equation*}
\Esp{\left(\vct Z_1^\ast \vct Z_1\right)^{-1}}=\sum_{p=1}^P (-\lambda_p)^{P-1}\frac{\mathrm{ln} \lambda_p}{\lambda_p}\prod_{p^\prime\neq p}\left( \lambda_{p^\prime} -\lambda_p\right)^{-1}
\end{equation*}
\end{itemize}
\end{proposition}\begin{proof}See \cite{Barbotin2010}. The uncorrelated case is found in various statistical handbooks as the first moment of an inverse-$\chi^2$ distributed random variable. For the correlated case, see \cite{Thomas1976}.\end{proof}
This expression is a suitable approximation for multipaths  scenario with distant paths (separated by more than twice the inverse bandwidth.
It gives an important insight on the evolution of the estimation performance when uncorrelated antennas are added to the system. Namely, the RMSE decays as $1/\sqrt{P-1}$.

In general,  multiple paths are interacting with each other and the information matrix cannot be considered diagonal. I	n this case Yau and Bresler \cite{Yau1992} derived the following expression:
\begin{proposition}\cite{Yau1992}\label{prop:nsep}
Let $\mat \Phi$ and $\mat \Phi^\prime$ be $N\times K$ matrices such that
$$\Phi_{n,k}=\varphi((n-1)T-t_k)\ ,\quad\Phi^\prime_{n,k}=\varphi\prime((n-1)T-t_k),$$
$n\in\lbrace 1,\dots,N \rbrace,\ k\in\lbrace 1,\dots,K \rbrace$. Given the stochastic matrix
$$\mat{C}= \text{diag}(a_1,\dots,a_K)\left(\sum_{p=1}^P\vct Z^\prime_p\vct Z^{\prime\ast}_p\right)\text{diag}(a_1^\ast,\dots,a_K^\ast),$$
with $\vct Z^\prime_p=[Z_{1,p}\cdots Z_{K,p}]^T$,
the Fisher information matrix $\mat J$ conditionned on the path amplitudes is given by
\begin{equation}
\mat J=2\sigma^{-2}\mat{\Phi}^{\prime\ast}P_{\text{ker}\mat{\Phi}}\mat{\Phi}^\prime\odot \mat{C}.
\end{equation}
such that $P_{\text{ker}\mat{\Phi}}=\mathbb{I}-\mat\Phi\mat\Phi^\pinv$ is the projection into the nullspace of $\mat\Phi$ and ``$\odot$'' denotes the entrywise matrix product.
\end{proposition}

See \cite{Yau1992} for the proof. The Cram\'er-Rao bounds  for the estimation of the normalized times of arrival are on the diagonal of the expectation of $\mat J^{-1}$.
The matrix $\mat J$ is a complex Wishart matrix. Computing its inverse moments analytically is not an easy task, nevertheless it can be numerically computed via Monte-Carlo simulations.

\section{Application to OFDM and CDMA downlink}\label{sec:ofdm}
\subsubsection{SCS-FRI with uniformly scattered DFT pilots (OFDM)}
The theory in Section \ref{sec:theory} is developed for contiguous DFT coefficients. In OFDM communications, pilots are often uniformly laid out in frequency (ETSI DVB-T \cite{ETSI2001}, 3GPP LTE \cite{Dahlman2008},\dots). The period of pilot insertion $D$ is upper-bounded by $\Delta^{-1}$, the inverse of the delay-spread of the CIR : $D<\tau/\Delta$.
If  not, the CIR cannot be unambiguously recovered from the pilots because of aliasing. For a fixed number of pilots, $D$ is chosen as large as possible ($D=\lfloor\tau/\Delta\rfloor$), as interpolation of the CIR spectrum is more robust than extrapolation.

SCS-FRI can take advantage of uniformly scattered pilot layouts \cite{Maravic2004,Barbotin2010}. For $\widehat\varphi$ flat in $\lbrace-MD,\dots, MD\rbrace$, equation (\ref{eq:spldft}) becomes:
\begin{equation}\widehat{y}_p[mD]=\sum_{k=1}^K c_{k,p}W_N^{mDt_k}+\widehat{q}_p[mD],\label{eq:splofdm}\end{equation}
which corresponds to a dilation by $D$ of the support parameters $\lbrace t_k\rbrace$. By definition $0\leq t_k < \Delta$, and so the bound on $D$ prohibits aliasing of $Dt_k$. Therefore, SCS-FRI is applicable without other modification than division of the recovered support parameters by $D$.
The results of Proposition \ref{prop:crbp} can be extended to scattered pilot with minimal effort.
\begin{corollary}
\label{thm:ofdmCRB}
The minimal uncertainties on the estimation of the parameters in the SCS-FRI scenario (\ref{eq:splofdm}) with $P$ signals are given by
\begin{align*}
  \Esp{\left(\frac{\vartriangle\!\! t_1}{\tau}\right)^2} &\geq \frac{3BT}{4D^2 \pi^2 M(M+1)}\:\Esp{\mathsf{ESNR}^{-1}} \\
  \Esp{\left(\frac{\vartriangle\!\! c_{\ell}}{c_{\ell}}\right)^2}&\geq BT\:\Esp{\mathsf{PSNR}_{\ell}^{-1}}  \hspace{0.5 cm}     \ell = 1, \ldots, P.
\end{align*}
For real-valued signal and noise  $\mathsf{ESNR} = \frac{1}{\sigma^2}\sum_{\ell = 1}^{P} c_{\ell}^2$ denotes the effective signal to noise ratio and $\mathsf{PSNR}_{\ell} = c_{\ell}^2/\sigma^2\,.$. For complex-valued signal and noise  $\mathsf{ESNR} = \frac{1}{2\sigma^2}\sum_{\ell = 1}^{P} c_{\ell}^\ast c_{\ell}$ and $\mathsf{PSNR}_{\ell} = c_{\ell}^\ast c_{\ell}/(2\sigma^2)\,.$.
\end{corollary}
\begin{proof}
$$\Esp{\left(\frac{\vartriangle\!\! t_1}{\tau}\right)^2}=\Esp{\left(\frac{\vartriangle\!\! (Dt_1)}{\tau}\right)^2}\cdot D^{-2}.$$
Evaluation of $\Esp{\left(\frac{\vartriangle\!\! (Dt_1)}{\tau}\right)^2}$ based on measurements from (\ref{eq:splofdm}) is answered by Proposition \ref{prop:crbp}. The \textsf{dSNR} is explicitly computed for (\ref{eq:splofdm}) taking $\varphi=\text{sinc}_B$. \end{proof}
\subsubsection{Extension to Walsh-Hadamard coded schemes (CDMA)}
Numerous applications use the $2^n$-WHT to code the channel into $2^n$ subchannels ($N=2^n$). Among others, IS-95 uses a $64$-WHT to code the downlink channel. The straightforward way to insert pilots is to use one of these subchannels as a pilot itself and use correlation based channel estimation methods as the Rake-receiver for example  \cite{Price1958}. The SCS-FRI algorithm works in the DFT domain but can nevertheless be applied in the WHT domain with pilots uniformly scattered by $D$ a power of 2.
\begin{proposition}\label{prop:WHT}
Let $\mat W_n$ and $\mat S_n$ be respectively the $2^n$-points DFT and WHT matrices obtained by Sylvester's construction:
$$\mat S_1=\frac{1}{\sqrt{2}}\begin{bmatrix}1 &1\\1 &-1\end{bmatrix},\quad \mat S_{i+1}=\mat S_1 \krop \mat S_i.$$
Then, for any $\ell\in\lbrace 1,\dots , n-1\rbrace$ the set of $\mat S_n$'s columns  with indices in $\left\lbrace 2^\ell+i\right\rbrace_{1\dots 2^\ell}$  and the set of $\mat W_n$'s columns with indices in $\left\lbrace (i-1/2)\cdot 2^{n-\ell}+1\right\rbrace_{1\dots 2^\ell}$ span the same subspace.
\end{proposition}
\begin{proof}
We partition the Walsh-Hadamard transform matrix in two ``left'' and ``right'' blocks:
$$\mat S_n=\begin{bmatrix}\mat S_n^{(l)}& \mat S_n^{(r)}\end{bmatrix},\ \ \mat S_n^{(l)}=\begin{bmatrix}\mat S_{n-1}\\ \mat S_{n-1}\end{bmatrix},\ \ \mat S_n^{(r)}=\begin{bmatrix}\mat S_{n-1}\\ -\mat S_{n-1}\end{bmatrix}.$$
Given $\vct w_{2^n}^k=[W_{2^n}^{0k}\cdots W_{2^n}^{(2^n-1)k}]$ the $k^{th}$ vector of the DFT basis and $\vct s^{(r)}\in \text{span }{\mat S_n^{(r)}}$:
\begin{align*}
\left\langle\vct{w}_{2^n}^k,\vct{s}_n^{(r)}\right\rangle &= \sum_{l=0}^{N-1}W_{2^n}^{kl}s_n^{(r)}[l] \\
 &= \sum_{l=0}^{2^{n-1}-1}W_{2^n}^{kl}s_n^{(r)}[l]+W_{2^n}^{k(l+2^{n-1})}s_n^{(r)}[l+2^{n-1}]\\
 &= \sum_{l=0}^{2^{n-1}-1}W_{2^n}^{kl}s_n^{(r)}[l]\left(1-W_{2^n}^{k2^{n-1}}\right)
\end{align*}
Hence
$$\left\langle\vct{w}_{2^n}^k,\ \vct{s}_n^{(r)}\right\rangle=0\ ,\quad \text{for $k$ even.}$$
The spans of $\mat S_n^{(r)}$ and $\mat S_n^{(r)}$ partition the original $2^n$-dimensional space in two subspaces of  dimension $2^{n-1}$. Let $\mat W_n^{(o)}=\lbrace\vct w_{2^n}^k\rbrace_{0 \leq 2k+1<N,k\in\Nspace}$ be the DFT basis vectors with odd indices and $\mat W_n^{(e)}=\lbrace\vct w_{2^n}^k\rbrace_{0 \leq 2k<N,k\in\Nspace}$ the ones with even indices. The spans of $\mat W_n^{(o)}$ and $\mat W_n^{(e)}$ partition the original space into two subspaces of dimension $2^{n-1}$. Since $\text{span }\mat S_n^{(r)}\ \perp\ \text{span }\mat W_n^{(e)}$:
\begin{align*}\text{span }\mat W_n^{(o)} &= \text{span }\mat S_n^{(r)},\\
\text{span }\mat W_n^{(e)} &= \text{span }\mat S_n^{(l)}.
\end{align*}
This property applies recursively, since for $k\in\lbrace0,\dots, 2^{n-1}-1\rbrace$:
\begin{align*}\left\langle\vct{w}_{2^n}^{2k},\ \vct{s}_n^{(l)}\right\rangle &=\left\langle\frac{1}{\sqrt{2}}\begin{bmatrix}\vct{w}_{2^{n-1}}^k\\\vct{w}_{2^{n-1}}^k\end{bmatrix},\ \frac{1}{\sqrt{2}}\begin{bmatrix}\vct{s}_{n-1}\\\vct{s}_{n-1}\end{bmatrix}\right\rangle,\\
&=\left\langle\vct{w}_{2^{n-1}}^{k},\ \vct{s}_{n-1}\right\rangle,
\end{align*}
where $\vct s_{n-1}\in \text{span }{\mat S_{n-1}}$.
\end{proof}
Proposition \ref{prop:WHT} states that one can choose $2^\ell$ contiguous Walsh-Hadamard codewords for pilots and get $2^\ell$ uniformly spread DFT pilots with layout gap $D=2^{n-\ell}$. The channel coding is akin CDMA, but the pilot layout matches the one used in OFDM communication. The lesson, is that the Walsh-Hadamard transform alone achieves ``scrambling'' of data followed by carrier mapping in the DFT domain in a fashion similar to SC-FDMA \cite{Myung2007}. In SC-FDMA, the data are first ``scrambled'' by application of a shorter length DFT.

This result  has a nice interpretation in the context of generalized Fourier transforms, the $2^n$-WHT being itself the Fourier transform on the finite group $(\Zspace/2\Zspace)^n$ instead of $\Zspace/2^n\Zspace$ for the classical $2^n$-points DFT \cite{Barbotin2010,Terras1999}. A similar result holds for DFT on any toric finite group \cite{Barbotin2010}.

\section{Application: Fading channel estimation in multi-output systems}\label{sec:application}
\subsection{Channel model}
\subsubsection{Physical assumptions}
	A linear time-invariant channel is characterized by its impulse response $h$. In mobile communications, channels are transient, but we may assume the channel to be locally invariant around a time $\tau$. This leads to the definition of a time-dependent channel impulse response $h_\tau$. Consider a channel impulse response made of a large number $L$ of echoes:
	\begin{equation}\label{eq:echoChannel}
	h_\tau(t)=\sum_{l=1}^L \alpha_l(\tau)\delta(t-t_l(\tau)).
	\end{equation} 
	The number of echoes $L$, is usually far too large to warrant a finite rate of innovation approach. However individual echoes aggregate in a smaller and manageable number of clusters $K$ \cite{Turin1956}. The rationales behind clustering are the same as for the common support assumption: a finite bandwidth combined with background noise allow only for a limited resolution. Table~\ref{tabl:invB} lists a few examples for which clustering applies in typical operating conditions.

	This simplification is at the heart of medium and narrowband wireless communications \cite{Molisch2005}. We want to estimate $h_\tau$ by sending probes at the input and collecting samples at the output.
	
	 Correlation of the channel with respect to time is an important feature to exploit, however we will not consider it, as scheduling in modern communication systems makes its usage uncertain. Hence we settle on a time $\tau$ and drop it from the notation.
	 
	 Communication is carried over a restricted frequency band, which is achieved by pulse-shaping with a template function $\varphi(t)$ and modulation by $e^{j\omega_c t}$. Applying clustering to (\ref{eq:echoChannel}) the channel impulse response becomes:
	 \begin{align}
	 h(t)&=\sum_{k=1}^K c_k\varphi(t-t_k)\label{eq:mpchannel}\\
	 \text{s.t. }c_k&=a_kZ_k= e^{j\omega_c t_k}\sum_{(\alpha_l,t_l)\in\mathcal{C}_k}\alpha_le^{j\omega_c(t_l-t_k)},
	 \end{align}
	 where $Z_k$ has unit-variance, $a_k$ is the appropriate scaling parameter and $\mathcal{C}_k$ is the $k^{th}$ cluster. Assuming $\lbrace \alpha_l e^{j\omega_c(t_l-t_k)}\rbrace_{(\alpha_l,t_l)\in\mathcal{C}_k}$ contains i.i.d. elements with finite first two moments (echoes of finite energy)
	 
	 $$\lim_{\#\mathcal{C}_k\rightarrow\infty} Z_k \sim \mathcal{N}_{\mathbb C}(0,1).$$
	 
	 This is the classical non line of sight fading scenario where the paths amplitudes $|Z_k|$ are independently Rayleigh distributed.
	 \subsubsection{Multipoint communications, one to many}\label{sec:spmod}
	Communication through fading channels rely on spatial diversity to gain robustness. Spatial diversity is achieved with the deployment of several antennas at the receiver and/or transmitter. We describe a spatial channel model between one transmitter and several receivers, which generalizes to MIMO communications in a straightforward manner. The physical properties of the channel are the following as shown in Figure~\ref{fig:mpchannel}:
	\begin{itemize}
	\item The distance in between antennas $m$ and $n$ is $d_{m,n}$.
	\item Each path is characterized by an angle of arrival (AoA) $\theta_k$. To simplify computations it is assumed that the AoA is the same for all antennas (far field assumption). In the near field, a scatterer surrounds the receiver and the distribution becomes almost isotropic. Hence this assumption can be made for both regimes with limited error.
	\item The direction normal to the segment between antennas $m$ and $n$ points toward azimuth $\theta_{m,n}$.
	\end{itemize}
	\begin{figure}[h]
\centering 
\includegraphics[width=3.3in]{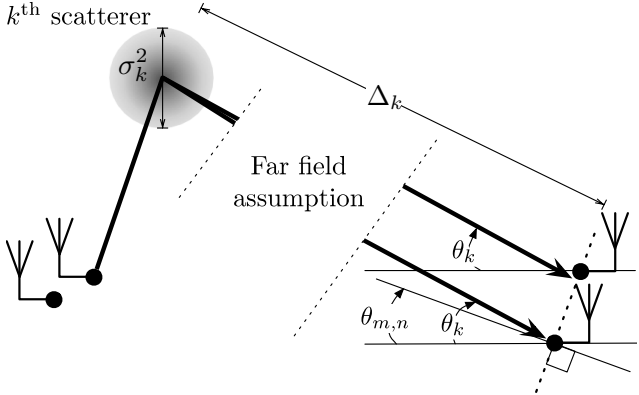}
\caption{Channel model with a single scatterer. Each scatterer is characterized by its apparent width (width/distance) $\sigma_k/\Delta_k$ and its azimuth $\theta_k$. This model is considered valid in the near field as well, as the scatterer surrounds the receiver, thus having no intrinsic azimuth.}
\label{fig:mpchannel}
\end{figure}
	The channel model in (\ref{eq:mpchannel}) applies to the $P$ subchannels
	\begin{equation}
	  h_m(t)=\sum_{k=1}^K a_{m,k}Z_{m,k}\varphi(t-t_{m,k})\ ,\quad m=1,\dots,P.\\
	\end{equation}
	We assume the distance in between antennas is smaller than the achievable spatial resolution, hence $$t_{1,k}=t_{2,k}=\dots=t_{P,k}.$$
	To fully characterize the channel, the path correlation across antennas must be known --- by assumption $Z_{m,k}$ and $Z_{m^\prime,k^\prime}$ are independent for $k\neq k^\prime$.
	Following Salz and Winters \cite{Salz1994} we derive a formula for the autocorrelation matrix of $\vct{Z}_k=[Z_{1,k} \cdots Z_{P,k}]$. However we put a Gaussian prior on the cluster shape rather than a uniform one with discontinuities at the boundaries. 
\subsubsection{Spatial correlation of paths}
As in \cite{Salz1994}, a large number of reflections are assumed to be drawn from a continuous probability distribution for each scatterer. 
\begin{proposition}\label{prop:spacecorr}
Under the spatial channel model described in Section~\ref{sec:spmod}, the antenna crosscorrelation is closely approximated by:
\begin{align}
\Esp{Z_{k,m}Z_{k,n}^\ast}=&\quad J_0\left(\frac{\omega_c}{c}d_{m,n}\right)\nonumber\\&+2\sum_{l=1}^\infty j^{l}\frac{I_l(\kappa_k)}{I_0(\kappa_k)}J_l\left(\frac{\omega_c}{c}d_{m,n}\right)\\
&\qquad\qquad\cdot\cos\left[l\left(-\theta_{m,n}+\theta_k-\frac{\pi}{2}\right)\right],\nonumber
\end{align}
where $\Delta_k^2/\sigma_k^2\approx (1-e^{-3\kappa_k/4})\kappa_k$, $J_{(\cdot)}$ is the Bessel function of the first kind and $I_{(\cdot)}$ is the modified Bessel function of the first kind.
\end{proposition}
\begin{proof}
See appendix \ref{sec:appendixB}. It is only a close approximation since the azimuthal distribution at the receiver is approximated by a Von-Mises distribution.
\end{proof}
\begin{corollary}
For a path width $\kappa_k$ large enough:
\begin{align}
\Esp{Z_{k,m}Z_{k,n}^\ast}\approx &\quad J_0\left(\frac{\omega_c}{c}d_{m,n}\right)\nonumber\\
&+2\sqrt{2\pi\kappa_k}\sum_{l=1}^\infty j^{l}f_{\kappa_k}(l)J_l\left(\frac{\omega_c}{c}d_{m,n}\right)\\
&\qquad\qquad\quad \cdot\cos\left[l\left(-\theta_{m,n}+\theta_k-\frac{\pi}{2}\right)\right],\nonumber
\end{align}
such that $f_{\kappa_k}$ is the centered Gaussian pdf of variance $\kappa_k$.
\end{corollary}
The result is in its form similar to \cite{Salz1994}, however the derivation stays closer to the original physical model. 

\section{Numerical results}
\begin{table}[t]\centering \ra{1.3}\caption{Simulation parameters}
\begin{tabular}{@{}lcl@{}}\toprule
\textbf{Parameter}& \textbf{Symbol}& \textbf{Value}\\\midrule
 Sampling step &$T$&$50\mathrm{ns}$\\ 
 Bandwidth & $B$& $20\mathrm{MHz}$\\
 Center frequency & $f_c$ & $2.6\mathrm{GHz}$\\ 
 Frame duration (without padding)& $\tau$& $25.55\mathrm{\mu s}$\\
 Samples per frame &$N_{\text{frame}}$& $511$\\
 Pilots per frame &$N$&$63$\\
 Pilot gap &$D$&$8$\\
 Delay spread & $\Delta$ & $1.6\mathrm{\mu s}$\\\bottomrule
\end{tabular}
 \label{tabl:rxspec}
\end{table}

For simulations we use the channel model developed in Section~\ref{sec:application}, and choose its parameters to \emph{loosely} follow the 3GPP-LTE standard. Its characteristics are listed in Table~\ref{tabl:rxspec}. We assume 63 pilots which are uniformly spaced in frequency, one every 8.  The transmitted frame is circularly padded such as to guarantee circular convolution of the transmitted signal with the CIR. 
Results are derived from three different experiments:
\begin{itemize}
\item[\textsf A] The medium has two paths separated by $2T$. The second path's expected amplitude is $1/10^{th}$ of the expected amplitude of the first path. The receiver possesses 1, 2, 4 or 8 uncorrelated antennas. The channels have exact SCS ($\varepsilon=0$).
\item[\textsf B] The medium has two paths separated by $T$ or $2T$. Both paths have the same expected amplitude. The receiver has 4 uncorrelated antennas. The channels have either exact SCS ($\varepsilon=0$) or non-exact SCS ($\varepsilon=T/50=1$ns). The discrepancy in the ToA between antennas is uniformly distributed in $[-\varepsilon\ \varepsilon]$. A time lapse of $2T/50$ corresponds to a path length difference of $60$~cm. 

\begin{figure*}[ht!]
  \centering
\includegraphics[width=0.82\textwidth]{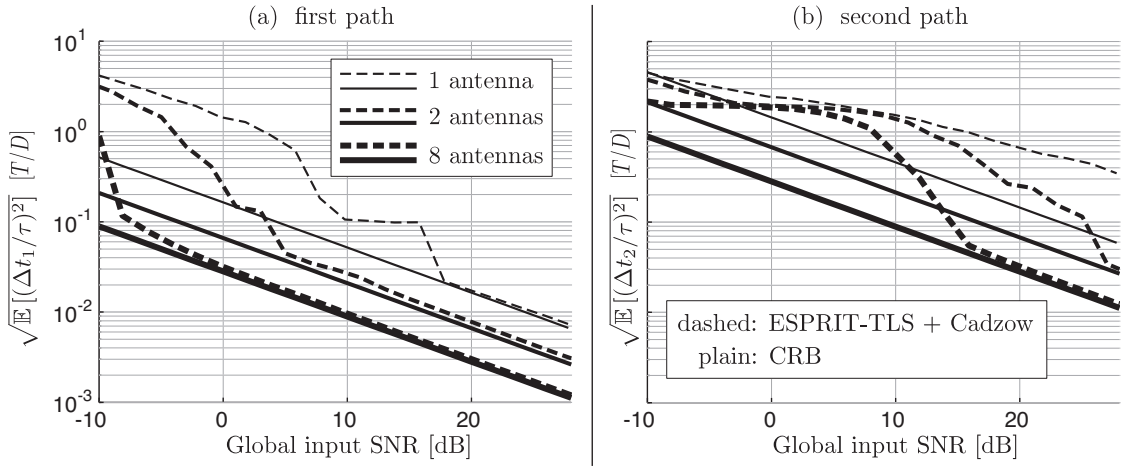}
\caption{ (\xp A) For the same global input SNR, a system with more antennas estimates the ToAs more accurately and is more resilient to noise. This is a consequence of the increased receiver diversity. The second path has $1/10^\text{th}$ the amplitude of the first path and is thus quickly buried into noise as SNR decreases. The estimation reaches the Cram\'er-Rao bound as long as it correctly identifies the path.}
\label{fig:fig1}
\end{figure*}

\begin{figure*}[ht!]
  \centering
\includegraphics[width=0.82\textwidth]{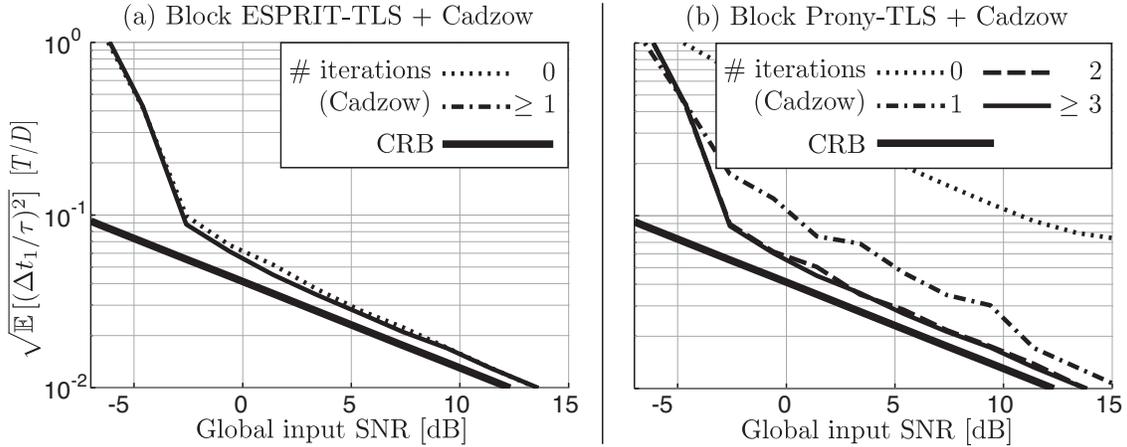}
\caption{ (\xp A) Part (a) shows the performances of Block ESPRIT-TLS with or without Block Cadzow denoising. In this setup, the gain obtained with the denoising is relatively small and is achieved after one iteration. Part (b) shows the performances of Block Prony-TLS with or without Block Cadzow denoising. As expected, the performance of Prony's algorithm without denoising is very poor. After 3 denoising iterations, performances of Block Prony-TLS and Block ESPRIT-TLS are indistinguishable.
}
\label{fig:fig2}
\end{figure*}

\item[\textsf C] This experiment is more realistic from a physical standpoint. The receiver has 5 antennas equispaced on a circle of radius $10$~cm. The propagation medium contains 4 scatterers (Figure~\ref{fig:chan}.(a)). The expected CIR modulus is represented in Figure~\ref{fig:chan}.(b).  We use the spatial correlation model derived in Proposition~\ref{prop:spacecorr}, and provide the antennas cross-correlation in Figure~\ref{fig:chan}.(c). Also the channel is not exactly SCS, with a maximum delay $\varepsilon=T/50=1$ns.
\end{itemize}
Results were obtained on 400 independent noise and fading realisations.
\begin{figure*}[ht!]
  \centering
\includegraphics[width=0.8\textwidth]{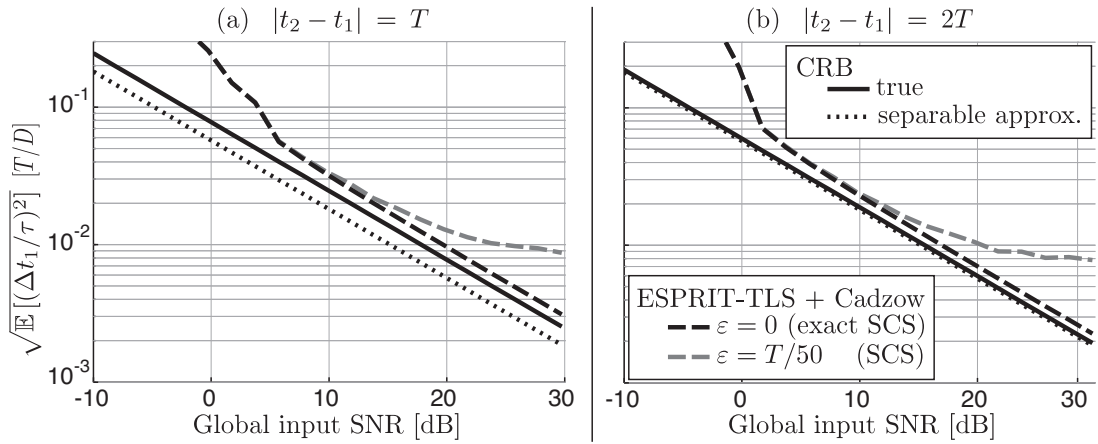}
\caption{ (\xp B) This figure shows that the proposed algorithms behave as expected in the presence of ToA mismatches between antennas. Part (b) motivates the separability assumption to compute the CRB of paths located more than $2T$ apart, while Part (a) shows its inadequacy for a smaller delay $T$. The ``true'' estimate is obtained via Monte-Carlo simulations.}
\label{fig:fig3}
\end{figure*}

\begin{figure*}[ht!]
  \centering
\includegraphics[width=0.7\textwidth]{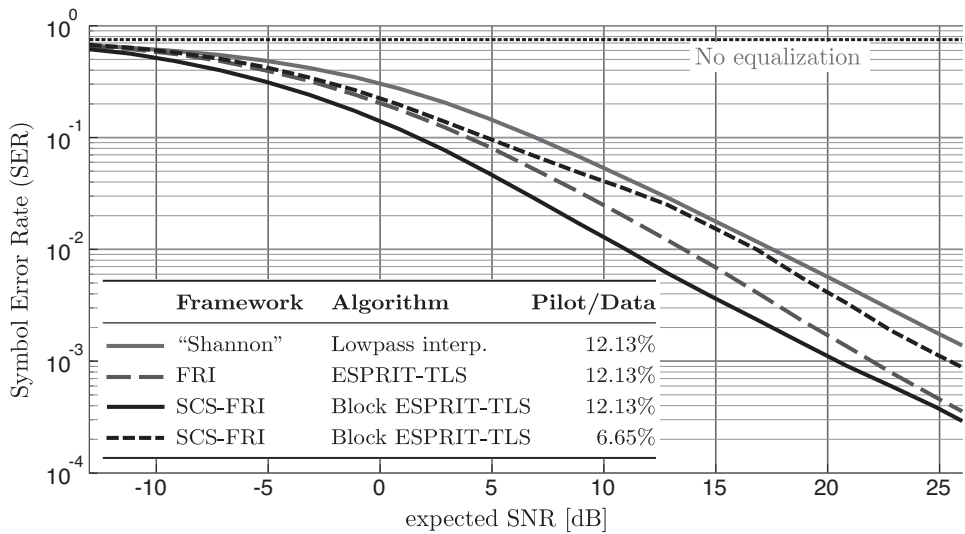}
\caption{ (\xp C) Using the SCS property, the SER is decreased by a factor $5$ above $10$dB of SNR compared to the conventional non-parametric approach. Sparsity alone provides a significant SER improvement, which shall be combined with the common support property below $30$dB of SNR. At very high SNR, independent channel estimation across antennas become preferable as the channels only approximately have the common support property. However, below $15$dB of SNR the effect of this approximation are undetectable. Another advantage of SCS-FRI is the reduction of pilots, it allows to halve their number while retaining performances superior to lowpass interpolation.}
\label{fig:fig4}
\end{figure*}
\subsubsection{Results on \xp A}
Figure~\ref{fig:fig1} shows that the SCS-FRI algorithm efficiently estimates the ToA down to a certain SNR where the recovery breaks down. This breaking point is pushed lower as spatial diversity increases, which is to be expected. Figure~\ref{fig:fig2} compares the use and combination of the various subspace identification techniques discussed earlier. The conclusion is that the performances of Block-ESPRIT TLS or Block-Prony TLS are exactly the same on a signal denoised with the Block-Cadzow algorithm. However Block-ESPRIT TLS requires fewer to none Block-Cadzow iterations than Block-Prony TLS to reach the optimum. It is well-known that Prony TLS is not robust to noise \cite{Blu2008,Tufts1982}.
\subsubsection{Results on \xp B}
Figure~\ref{fig:fig3} shows that the single path CRB given in Proposition~\ref{prop:crbp} is a good approximation of the true bound computed via Proposition~\ref{prop:nsep} for multiple paths separated by more than twice the inverse bandwidth of the channel. This experiment also verifies the usefulness of the SCS assumption when ToAs are slightly perturbed from one antenna to another: 
$$t_{k,p}=t_k+E_{k,p}\ ,\quad E_{k,p} \sim \mathrm U([-\varepsilon\ \varepsilon]),\ \text{i.i.d.}$$
The error caused by the random perturbation $E_{k,p}$ is of the order of the perturbation itself, and thus we may say SCS-FRI is robust on non exact SCS channels.
\subsubsection{Results on \xp C}
All estimation algorithms use the fact that the delay spread is much shorter than the frame length. The difference between lowpass interpolation and other techniques is the use of the sparsity property. Using this property alone, the SER is halved at a SNR of $5$dB as shown in Figure~\ref{fig:fig4}. The addition of the SCS property proves to be valuable, at $5$dB of SNR the SER is decreased by a factor $3$. At high SNR, the SCS property provides a factor $5$ of improvement over lowpass interpolation. At very high SNR the error due to the approximate SCS nature of the channels diminishes this gain, and eventually the SCS assumption becomes detrimental.

It also shows that the number of pilots can be halved while having SER performances superior to the non-parametric approach (we retained half of the original pilots closest to the carrier frequency). For lowpass interpolation, this cannot be done without introducing aliasing. Reducing the number of pilots below ``Nyquist'' is relevant at high SNR where little redundency is required for denoising, leaving some additional spectrum for data transmission. In favorable transmission conditions, it would be possible to reduce the number of pilots down to the rate of innovation of the channel to maximize the data throughput.

\section{Conclusion}
We outlined the SCS-FRI algorithm, studied its performances on SCS channels estimation and computed theoretical lower-bounds for comparison. A spatial channel model was proposed for simulation purposes. The algorithm takes full advantage of the main properties of outdoor multipath channels and is directly applicable to most OFDM based communication standards. Simulations indicate that SCS-FRI based on the Block-ESPRIT TLS routine seems to be the most suitable since it requires only two partial SVD with size of the model order and provides optimal accuracy.

Future work is needed for estimation of the model order, incorporation of temporal correlation in the model and the algorithm (tracking of the model parameters). Computational complexity is also a crucial point for mobile applications, and improvements could be made with Krylov subspaces techniques as in \cite{Barbotin2009}.


%

\appendices

\section{Spatial correlation formula for fading channels}\label{sec:appendixB}

\subsection{Azimuthal scatterers density distribution}
The reflection density of each scatterer is normally distributed with mean $\vct\mu_k$ (its position) and covariance matrix $\sigma_k^2\mathbb I$ (its ``girth''). The number of reflections within a scatterer is assumed to be large enough to warrant their approximation by their continuous probability density function. The azimuthal  density is the integral of the scatterer's pdf over $\Gamma_\vartheta$ the straight path from the receiving antenna at an angle\footnote{Without loss of generality the scatterer origin is at azimuth $0$, and the antenna is located at position $\vct 0$} $\vartheta$:
\begin{equation}
p(\vartheta;\vct \mu_k, \sigma_k^2)=\int_{\Rspace^2}f^{(2D)}_{\sigma_k^2}(\vct x -\vct\mu_k)\mathcal I_{\vct x\in\Gamma_\vartheta}\vct{dx}.
\end{equation}
Reparametrization in polar coordinates yield:
\begin{align*}
p(\vartheta;\vct \mu_k, \sigma_k^2)=\: &f_{\sigma_k^2}(\|\vct\mu_k\|\sin(\vartheta))\\ &\cdot\int_{\Rspace_+}f_{\sigma_k^2}(r-\|\vct\mu_k\|\cos(\vartheta))J_{\vct x}(r,\vartheta)dr,\\
=\: &\sigma_k^{-1}f\left(\sqrt{\kappa^\prime_k}\sin(\vartheta)\right)\\ &\cdot\int_{\sqrt{\kappa^\prime_k}\cos(\vartheta)}^{+\infty}\sigma_k^{-1}f(r-\|\vct\mu_k\|\cos(\vartheta))\\ & \qquad\qquad\qquad\cdot (s+\sqrt{\kappa^\prime_k}\cos(\vartheta))\sigma_k^2 ds
\end{align*}
such that $\kappa^\prime_k=\|\vct\mu_k\|^2/\sigma_k^2$ and $J_{\vct x}(r,\vartheta)=r$ is the Jacobian of the cartesian to polar transformation. We performed the change of variable $s=r-\sqrt{\kappa^\prime_k}\cos(\vartheta)$.
Hence, the distribution has only one degree of freedom, and after some calculus:
\begin{align}\label{eq:circdist}
p_{\kappa^\prime_k}(\vartheta)=\: &f(\sqrt{\kappa^\prime_k}\sin\vartheta)\\ &\cdot\left[\sqrt{\kappa^\prime_k}\cos\vartheta\cdot F(\sqrt{\kappa^\prime_k}\cos\vartheta)f(\sqrt{\kappa^\prime_k}\cos\vartheta)\right].\nonumber	
\end{align}
The circular distribution (\ref{eq:circdist}) is well approximated by a Von-Mises distribution of scale $\kappa_k$:
\begin{equation}
q_{\kappa_k}(\vartheta)=\frac{e^{\kappa_k\cos\vartheta}}{2\pi I_0(\kappa_k)}.
\end{equation}
where $I_0$ is the $0^{th}$ order modified Bessel function of the first kind. Asymptotically, $\kappa^\prime_k\overset{\kappa^\prime_k\rightarrow \infty}{=}\kappa_k$, and the approximation $\kappa^\prime_k\approx (1-e^{-3\kappa_k/4})\kappa_k$ was found to be empirically accurate for all $\kappa_k$ (K-L divergence between $p_{\kappa^\prime_k}$ and $q_{\kappa_k}$ is less than $0.02$ bits). 
\subsection{Derivation of the correlation matrix formula}
Considering the setup of Figure~\ref{fig:mpchannel}, and from \cite{Salz1994}:
$$R_Z^{(k)}[m,n]\:=\:\int_{-\pi}^{\pi} q_{\kappa_k}(\vartheta+\theta_{m,n}-\theta_k) e^{j\frac{\omega_c}{c}d_{m,n}\sin\vartheta}d\vartheta.$$
Then, $q_{\kappa_k}$ is expanded in terms of spherical harmonics via the Jacobi-Anger expansion \cite{Abramowitz1964}(9.1): 
\begin{align*}& q_{\kappa_k}(\vartheta+\theta_{m,n}-\theta_k)\\=\: &\frac{1}{2\pi I_0(\kappa_k)}\left\lbrace J_0(-j\kappa_k)\vphantom{\sum_{l=1}^\infty}\right.\\
&\left.+\sum_{l=1}^\infty j^lJ_l(-j\kappa_k) \cos\left[l(\vartheta+\theta_{m,n}-\theta_k)\right]\right\rbrace\ ,\\
=\: &\frac{1}{2\pi}+\frac{1}{\pi I_0(\kappa_k)}\sum_{l=1}^\infty I_l(\kappa_k) \cos\left[l(\vartheta+\theta_{m,n}-\theta_k)\right]\ ,
\end{align*}
where the second equality is obtained with $I_l(jx)=j^lJ_l(x)$ \cite{Abramowitz1964}(9.6.3, 9.1.35).

We now have a series for $R_Z^{(k)}[m,n]$ with $l^{th}$ term:
\begin{align*} & \frac{ I_l(\kappa_k)}{\pi I_0(\kappa_k)}\int_{-\pi}^{\pi}\cos\left[l(\vartheta+\theta_{m,n}-\theta_k)\right] e^{j\frac{\omega_c}{c}d_{m,n}\sin\vartheta}d\vartheta \\
\overset{(a)}{=}\quad&\frac{ I_l(\kappa_k)}{\pi I_0(\kappa_k)}
\left\lbrace\vphantom{\int_{-\pi}^{\pi}}\cos\left[l\left(\theta_k-\theta_{m,n}-\frac{\pi}{2}\right)\right]\right.\\ 
&\qquad\qquad\qquad\cdot\int_{-\pi}^{\pi}\cos l\vartheta\; e^{j\frac{\omega_c}{c}d_{m,n}\cos\vartheta}d\vartheta\\
 & \qquad\qquad+\sin\left[l\left(\theta_k-\theta_{m,n}-\frac{\pi}{2}\right)\right]\\
 &\qquad\qquad\qquad\left.\cdot\int_{-\pi}^{\pi}\sin l\vartheta\; e^{j\frac{\omega_c}{c}d_{m,n}\cos\vartheta}d\vartheta\;\right\rbrace\\
\overset{(b)}{=}\quad& \frac{2 I_l(\kappa_k)}{I_0(\kappa_k)}I_l\left(j\frac{\omega_c}{c}d_{m,n}\right)\cos\left[l\left(\theta_k-\theta_{m,n}-\frac{\pi}{2}\right)\right]\\
\overset{(c)}{=}\quad& \frac{ 2j^{l}I_l(\kappa_k)}{I_0(\kappa_k)}J_l\left(\frac{\omega_c}{c}d_{m,n}\right)\cos\left[l\left(\theta_k-\theta_{m,n}-\frac{\pi}{2}\right)\right]
\end{align*}
Equality $(a)$ is obtained with some standard trigonometric identities and a shift by $-\frac{\pi}{2}$ of the variable of integration. Equality $(b)$ follows from the standard integral representation of $I_l$ (\cite{Abramowitz1964} 9.6.19). The second integrand is antisymmetric which leads the integral over the unit-circle to vanish. Finally $(c)$ is a consequence of $I_l(jx)=j^lJ_l(x)$ again.
Hence:
\begin{align*}R_Z^{(k)}[m,n] =& J_0\left(\frac{\omega_c}{c}d_{m,n}\right)\\
&+\frac{2}{I_0(\kappa_k)}\sum_{l=1}^\infty j^{l}I_l(\kappa_k)J_l\left(\frac{\omega_c}{c}d_{m,n}\right)\\
&\qquad\qquad\qquad\cdot\cos\left[l\left(\theta_k-\theta_{m,n}-\frac{\pi}{2}\right)\right].
\end{align*}



\ifCLASSOPTIONcaptionsoff
  \newpage
\fi



\bibliographystyle{IEEEtran}
\bibliography{IEEEabrv,./SCS_FRI}
%

%

%
%
%




\end{document}